\documentclass[letterpaper,10pt,conference]{ieeeconf}

\IEEEoverridecommandlockouts
\usepackage{amsmath,amssymb,mathtools}
\usepackage{booktabs}
\usepackage{cite}
\usepackage{graphicx}
\usepackage{xcolor}
\usepackage{url}

\newtheorem{lemma}{Lemma}
\newtheorem{theorem}{Theorem}
\newtheorem{corollary}{Corollary}
\newtheorem{assumption}{Assumption}

\newcommand{\hcert}{h_{\mathrm c}}
\newcommand{\hsat}{h_{\mathrm{sat}}}
\newcommand{\dexact}{d^{\star}}
\newcommand{\qhat}{\widehat q}

\newcommand{\Oreal}{O^{\star}}
\newcommand{\Opred}{\widehat O}

\title{Barrier Function Conformal Safety Clearance Certification with CVaR for Driving Trajectory Selection}
\author{Pei Yu Chang$^{1}$ and Qadeer Ahmed$^{1}$
\thanks{*This research was supported by the CARMEN+ University Transportation Center, sponsored by the U.S. Department of Transportation under Grant No. 69A3552348327. The views presented are those of the authors and do not necessarily represent the official views of the U.S. Department of Transportation.}
\thanks{$^{1}$Pei Yu Chang and Qadeer Ahmed are the Department of Mechanical and Aerospace Engineering,The Ohio State University, Columbus, OH 43212 USA
        {\tt\small chang.2314, ahmed.358@osu.edu}}%
}

\begin{document}
\maketitle

\begin{abstract}

Autonomous driving motion planners generate and select candidate trajectories while accounting for interactions with surrounding agents. However, these evaluations do not certify the actual safety clearance of the selected trajectory. The framework evaluates the trajectory selected by ant planners and calibrates the gap between its plan time margin and realized safety clearance. A differentiable separating axis barrier margin deterministically lower bounds exact signed oriented-bounding-box (OBB) safety clearance, connecting the statistical certificate to safety margin. At plan time, the margin is evaluated using either a nominal prediction and sampled lower tail Conditional Value-at-Risk (CVaR), while post-selection conformal calibration over exchangeable drive sessions absorbs prediction and sampling errors. Conformal calibration provides statistical validity independently of predictor correctness. The method is evaluated on a frozen 300 session nuPlan study using native Predictive Driver Model (PDM) Closed loop proposals. At 10\% target miscoverage, sampled lower CVaR reduces the conformal correction from 1.43~m to 0.03~m and increases the rate of nonnegative safety clearance certificates from 68.7\% to 87.3\%. Across all evaluated statistics, exact-clearance coverage remains above the 90\% target at 93.3--96.7\%.

\end{abstract}

\section{Introduction and Related Work}

Autonomous driving in complex environments rarely uses a single trajectory, as ego behavior depends on uncertain interactions with surrounding agents. Motion planning has been a central component of autonomous driving, while recent advances in learned decision making and end-to-end driving have broadened trajectories are generated and evaluated. PDM-Closed builds route-aligned lateral paths and
scores them under longitudinal IDM policies~\cite{dauner2023pdm}; learned
planners emit several trajectory modes~\cite{fan2025riskaware,
zhang2025carplanner}; hybrid and end-to-end systems retain the same two-stage
shape while replacing parts of it with learned
components~\cite{fan2025sahdrivescenarioawarehybridplanner, song2025momad,
lin2025modelbasedpolicyadaptationclosedloop}.

Planners differ mainly in how candidates are scored and one trajectory is
chosen. PDM-Closed simulates each proposal and combines collision penalties with time-to-collision (TTC), route progress, and comfort~\cite{dauner2023pdm}. CarPlanner learns mode scores by reinforcement learning, then reranks with explicit safety and progress criteria~\cite{zhang2025carplanner}. Model-Based Policy Adaptation replaces the hand-designed score with a learned multi-step $Q$-value~\cite{lin2025modelbasedpolicyadaptationclosedloop}, and SAH-Drive arbitrates between rule-based and learned proposals without a common safety metric~\cite{fan2025sahdrivescenarioawarehybridplanner}. These quantities support trajectory ranking, but they do not directly certify the realized safety clearance of the selected trajectory.

Control barrier functions (CBFs) take a different route, defining safe sets and
enforcing forward invariance through constraints on the
dynamics~\cite{ALLAMAA2024392, chen2025dynamichighordercontrolbarrier,
RENGANATHAN2025472}. The guarantee is strong, but it is stated with respect to
assumed dynamics and to the environment representation available when the
constraint is imposed.
Surrounding vehicles behave in ways that are neither deterministic nor
unimodal, which makes a single nominal rollout an optimistic basis for
planning. Probabilistic formulations propagate state and sensing uncertainty into
the safety constraint, and tail-risk measures such as Conditional Value-at-Risk
(CVaR) weight the adverse part of the distribution~\cite{chang2026riskbudgeted,
CHANG2026107114}. Planners built this way evaluate a candidate against a
distribution of interactions~\cite{Mustafa2025riskaware}. However, a lower-tail statistic remains the statistic of the predicted distribution without certifying what is realized.

Beyond incorporating uncertainty into trajectory evaluation, recent work
verifies planned trajectories explicitly. Reachability methods
certify avoidance by separating the ego trajectory from the reachable regions
of surrounding agents~\cite{raeesi2025collision}. Conformal prediction supplies
finite sample, distribution free guarantees without assuming the predictor is
correct~\cite{vovk2005algorithmic}, and has been used
to build prediction regions and adapt safety margins for planning among dynamic
agents~\cite{lindemann2023safe, dixit2023adaptive}. Closer to our setting,
learned predictors have been calibrated conformally into forward reachable sets
that cover realized agent trajectories, after which the ego plan is checked
against those sets~\cite{chakraborty2025safety, pmlrv283contreras25a,
doula2025safepathconformalpredictionsafe}.

Existing certification methods leave two issues when applied to multimodal proposal planners. First, existing conformal planning methods primarily calibrate uncertainty in the predicted motion of surrounding agents, constructing prediction regions or reachable sets with probabilistic coverage~\cite{lindemann2023safe, dixit2023adaptive, chakraborty2025safety}. The resulting guarantees therefore concern whether the realized agent motion is contained within a calibrated uncertainty representation. Translating such coverage into a certificate on the physical clearance realized by an ego trajectory requires an additional geometric safety argument. Second, the ego trajectory to be certified is not fixed beforehand. Proposal-based planners select one trajectory from multiple candidates using predicted interactions, explicit safety criteria, or learned objectives~\cite{dauner2023pdm, zhang2025carplanner, fan2025sahdrivescenarioawarehybridplanner, lin2025modelbasedpolicyadaptationclosedloop}. The selection operation must therefore be included in the calibrated mapping if the certificate is intended to apply to the trajectory that is actually executed. The target of interest is thus the realized physical clearance of the selected trajectory, rather than prediction coverage or the safety score of an individual candidate.

We address these issues by directly calibrating the safety clearance of the trajectory selected by the planner. For each candidate trajectory, we construct a differentiable separating-axis margin that lower-bounds the exact signed oriented-bounding-box (OBB) safety clearance. The margin is evaluated over predicted agent futures using either a nominal estimate or a lower-tail CVaR statistic, after which the planner selects a trajectory. We then apply split-conformal calibration to the entire prediction, evaluation, and selection procedure, so that the resulting certificate applies to the selected trajectory rather than to an individual candidate. Rather than calibrating predictions or candidates separately, we apply conformal calibration to the complete selection procedure. The resulting certificate applies directly to the selected trajectory and provides a calibrated lower-clearance statement for a new exchangeable driving session.

The contributions are:
\begin{enumerate}
  \item \emph{Certified geometric safety margin.} Using standard separating axis and smooth min constructions, establishing a deterministic chain from a  differentiable safety margin to signed OBB safety clearance. This relation holds independently of the prediction model and ensures that a lower certificate on the smooth margin remains a valid lower certificate on physical vehicle safety clearance.
  \item \emph{Post-selection conformal certificate.} Calibrating the complete pipeline using split conformal prediction, so that the resulting safety clearance certificate applies to the trajectory actually selected by the planner.
  \item \emph{Risk-aware certificate tightness.} We use lower-tail CVaR to account for multimodal prediction uncertainty. CVaR affects certificate tightness and availability, while statistical validity is provided by conformal calibration.
  \item \emph{Large-scale empirical validation.} We evaluate the framework on 300 nuPlan driving sessions using PDM-Closed proposals~\cite{dauner2023pdm, karnchanachari2024nuplan}. At $10\%$ target miscoverage, CVaR reduces the conformal correction from $1.43$\,m to $0.03$\,m and increases nonnegative safety clearance certificates from $68.7\%$ to $87.3\%$, while maintaining $93.3$--$96.7\%$ held-out coverage.
\end{enumerate}

The proposed framework is intended for runtime trajectory evaluation rather than policy improvement. It attaches a calibrated safety clearance certificate to the trajectory selected by an existing planner, providing a safety assessment that can support downstream decisions and repeated evaluation during closed-loop operation. In our experiments, PDM-Closed serves as one application of this certification framework, with its original candidate generation and selection mechanism left unchanged.
\vspace{-4pt}

\section{Problem Formulation}

The planning information at each planning query occur at time \(t=0\) is
\[
I=(x^{\mathrm{ego}}_{[-T_{\mathrm h},0]},
  x^{1:J}_{[-T_{\mathrm h},0]},\mathcal M,\mathcal R),
\]
where \(x^{\mathrm{ego}}_{[-T_{\mathrm h},0]}\) is the ego vehicle state history, \(x^{1:J}_{[-T_{\mathrm h},0]}\) are the histories of the \(J\) tracked dynamic agents available at \(t=0\), \(\mathcal M\) is the local map context, and \(\mathcal R\) is the route context supplied to the planner. The history length \(T_{\mathrm h}\) and the state representation are determined by the external planner interface. 

Given \(I\), the planner generates a finite set of candidate ego trajectories \(\mathcal C(I)\). Each candidate \(\tau\in\mathcal C(I)\) contains \(K\) ego poses over a fixed horizon \([0,T]\). The fixed predictor, which may be part of the underlying planner, defines a conditional predictive distribution \(\widehat P(\cdot\mid I)\) over joint agent futures. Defining \(\Opred\sim\widehat P(\cdot\mid I)\) for a predicted joint future; in practice, the predictor may provide samples or a finite set for evaluating the candidates. In our experiments,
\(T=4\,\mathrm{s}\), \(K=41\), and the sampling interval is \(0.1\,\mathrm{s}\). The corresponding ground truth joint agent future \(\Oreal\) over \((0,T]\) is unavailable at planning time and is used only during calibration and evaluation.

For a candidate trajectory \(\tau\) and a joint agent future \(O\), let 
\(d_{obb}(\tau,O)\) denote the minimum signed geometric distance between the ego and agent OBBs over all evaluated time steps and agents. Positive values indicate separation, whereas negative values indicate penetration. The objective is the realized safety clearance \(d_{obb}(\tau,\Oreal)\). Since \(\Oreal\) is unavailable both when \(\tau\) is selected and when its certificate is computed, this quantity cannot be evaluated at planning time. Rather than forming a point estimate of the unknown realized safety clearance, our objective is to compute from information at plan time with a conservative lower certificate whose finite sample coverage is established by conformal calibration.

Conformal calibration is performed over drive session units rather than individual planning windows. Let \(U\) denote a session exchangeability unit and \(\mathcal W(U)\) its set of evaluated planning windows. All evaluated windows originating from the same session are assigned to the same unit, and their scores are  aggregated by the within unit maximum. The theory permits arbitrary dependence within a session and assumes exchangeability only across calibration and test units. In nuPlan experiment, each session contributes one evaluated planning query, so the unit maximum reduces to that query's score.

\begin{assumption}[Pipeline chosen before calibration]
\label{ass:pipeline}
The same procedure is used during calibration and deployment to produce the  selected trajectory and its margin. The procedure is chosen without using  calibration outcomes and does not use \(\Oreal\) at planning time. It may use internal randomness independent of \(\Oreal\).
\end{assumption}

\begin{assumption}[Deployment exchangeability]
\label{ass:exchangeability}
After applying the same window extraction and eligibility rules, the calibration sessions \(U_1,\ldots,U_n\) and a fresh session \(U_{n+1}\) from the target deployment population are exchangeable. Dependence among windows within the same session is unrestricted.
\end{assumption}

Assumption~\ref{ass:pipeline} ensures that calibration and deployment use the same procedure to produce the selected trajectory and its margin, excluding information unavailable at planning time and adaptation to calibration outcomes. Assumption~\ref{ass:exchangeability} provides the rank symmetry
required by split conformal prediction at the session level, while the unit maximum accounts for dependence among windows from the same session. These assumptions connect the calibration scores to a fresh deployment unit and enable the unit-level finite sample conformal coverage guarantee for realized safety clearance. Prediction and sampling errors in the plan time statistic are reflected in the conformal scores and therefore do not require a separate correctness guarantee for the statistic itself.




\section{Certified Geometric Margin}
\label{sec:certgeo}
Consider ego rectangle \(A\) and agent rectangle \(B\). Let \(e_s,e_d\) be the ego longitudinal and lateral unit axes, \(\Delta p\) the center displacement, and \(\Delta\psi\) the relative yaw.  For axis \(r\in\{s,d\}\), the signed projection gap is
\begin{equation}
h_r=|\Delta p^\top e_r|-R_{A,r}-R_{B,r}-\mu_r,
\label{eq:axis-gap}
\end{equation}
where \(R_{A,s}=L_A/2\), \(R_{A,d}=W_A/2\), and
\begin{align}
R_{B,s}&=\tfrac{L_B}{2}|\cos\Delta\psi|
       +\tfrac{W_B}{2}|\sin\Delta\psi|,\\
R_{B,d}&=\tfrac{L_B}{2}|\sin\Delta\psi|
       +\tfrac{W_B}{2}|\cos\Delta\psi|,
\end{align}
where \(L_A, L_B, W_A, W_B\)\ is the length and width of ego and agent rectangle. The nonnegative \(\mu_s,\mu_d\) are safety margin. The two axis SAT margin is \(\hsat=\max(h_s,h_d)\).  To obtain a smooth quantity without losing the
lower-bound direction, define
\begin{align}
h_{\mathrm{lse}}&=\alpha_g^{-1}\log
  \left(e^{\alpha_g h_s}+e^{\alpha_g h_d}\right),\\
\hcert&=h_{\mathrm{lse}}-\frac{\log 2}{\alpha_g},\qquad \alpha_g>0,
\label{eq:certmargin}
\end{align}
where parameter \(\alpha_g\)\ regulate the degree of smoothness. For trajectories, all margins denote the minimum over evaluated agent pairs. \cite{takasugi2024SSAT}

\begin{lemma}[Certified geometry]
For every candidate \(\tau\) and every joint future \(O\),
\begin{equation}
\hcert(\tau,O)\leq\hsat(\tau,O)\leq\dexact(\tau,O).
\label{eq:chain}
\end{equation}
\end{lemma}



\begin{proof} 
we have obtained \begin{equation} h_c \le h_{\mathrm{sat}} \end{equation} from~\eqref{eq:certmargin}. 
It remains to show $h_{\mathrm{sat}}\le d^\star$. It suffices first to consider the unpadded case $\mu_s=\mu_d=0$, since adding nonnegative padding can only decrease $h_s$ and $h_d$, and hence cannot increase $h_{\mathrm{sat}}$.

Suppose first that the two rectangles are disjoint. Along any unit axis $e_r$, $r\in\{s,d\}$, the projection gap in \eqref{eq:axis-gap} cannot exceed the Euclidean distance between the two rectangles. In particular, if $a^\star\in A$ and $b^\star\in B$ are a closest pair of points, then 
\begin{equation} 
h_r \le \left|e_r^\top(b^\star-a^\star)\right| \le \|b^\star-a^\star\|_2 = d^\star(A,B), 
\end{equation} 
where the second inequality follows from $\|e_r\|_2=1$. Therefore, \begin{equation} 
h_{\mathrm{sat}} = \max(h_s,h_d) \le d^\star(A,B). 
\end{equation}

Now suppose that the rectangles overlap. For two rectangles, the signed penetration distance is given by the largest signed projection gap over the complete set of SAT face-normal axes, 
\begin{equation} 
\mathcal{N}_{\mathrm{SAT}} = \{e_s^A,e_d^A,e_s^B,e_d^B\}. 
\end{equation} 
Hence, 
\begin{equation} 
d^\star(A,B) = \max_{n\in\mathcal{N}_{\mathrm{SAT}}} h_n, 
\end{equation} 
where $h_n$ denotes the corresponding unpadded signed projection gap. Since $h_{\mathrm{sat}}$ maximizes only over the two ego axes, 
\begin{equation} 
\{e_s^A,e_d^A\}\subseteq\mathcal{N}_{\mathrm{SAT}}, 
\end{equation} 
and therefore
\begin{equation} 
h_{\mathrm{sat}} \le \max_{n\in\mathcal{N}_{\mathrm{SAT}}} h_n = d^\star(A,B). 
\end{equation}

Restoring $\mu_s,\mu_d\ge0$ can only decrease $h_{\mathrm{sat}}$, so the inequality remains valid. Thus, for every pair, 
\begin{equation} 
h_c \le h_{\mathrm{sat}}\le d^\star. 
\end{equation}

Taking the minimum over all evaluated time steps and agents preserves both inequalities, yielding 
\begin{equation} 
h_c(\tau,O) \le h_{\mathrm{sat}}(\tau,O) \le d^\star(\tau,O). 
\end{equation} 
\end{proof}

The lemma is deterministic and independent with predictor. In the implementation,
the exact distance is used for retrospective verification, whereas
\(\hcert\) is the deployed evaluator margin.

\section{Risk-Aware Evaluation and Conformal Certification}
\label{sec:riskaware}
Section~\ref{sec:certgeo} establishes a deterministic geometric relation between the smooth separation axis margin and the signed OBB safety clearance for a given candidate trajectory and joint agent future. However, the actual future agent motion is unknown at planning time. We therefore evaluate the geometric margin under the predictive distribution of agent motion and calibrate the complete pipeline with observed future outcomes.
\vspace{-4pt}

\subsection{Plan-Time Risk Evaluation}

For a candidate trajectory $\tau$ and planning information $I$, let
\begin{equation}
M(\tau,I)
=
h_c(\tau,\widehat O),
\label{eq:random_margin}
\end{equation}
The predictive uncertainty in $\widehat O$ therefore induces a distribution over the certified geometric margin of each candidate trajectory.
A nominal evaluator uses a single predicted future $\widehat O_0$,
\begin{equation}
\widehat m_{\mathrm{nom}}(\tau,I)
=
h_c(\tau,\widehat O_0).
\label{eq:nominal_margin}
\end{equation}

To account for adverse predicted futures, we additionally use a lower-tail Conditional Value-at-Risk (CVaR). For
$\beta\in(0,1]$, define
\begin{equation}
\operatorname{CVaR}_{\beta}(M)
=
\frac{1}{\beta}
\int_0^\beta F_M^{-1}(u)\,du,
\label{eq:lower_cvar}
\end{equation}

where $F_M^{-1}$ denotes the quantile function of $M$. Because $M$ is a safety clearance for which smaller values are less safe, \eqref{eq:lower_cvar} averages the lower tail rather than the upper tail commonly used for loss variables. Smaller $\beta$ therefore places greater emphasis on adverse predicted futures.

In practice, the predictive distribution is represented by a finite set of predictor samples. We compute the empirical lower-tail CVaR from the corresponding margin samples and denote it by 
\begin{equation} 
\widehat m_{\mathrm{cvar}}(\tau,I) = \widehat{\operatorname{CVaR}}_{\beta} \left( M_1(\tau,I),\ldots,M_N(\tau,I) \right), 
\label{eq:sampled_cvar} 
\end{equation} 
where 
\begin{equation} 
M_i(\tau,I) = h_c(\tau,\widehat O_i) 
\end{equation} 
is the certified geometric margin under the $i$th predicted joint future. We use $\widehat m(\tau,I)$ to denote either $\widehat m_{\mathrm{nom}}(\tau,I)$ or $\widehat m_{\mathrm{cvar}}(\tau,I)$. Neither statistic is itself a realized safety clearance certificate; both are plan-time statistics whose prediction and sampling errors are absorbed by the post-selection conformal calibration described below.
\vspace{-4pt}

\subsection{Fixed Selection Pipeline for Calibration}

The trajectory to be certified is not fixed in advance, but is selected using information available at planning time. Because this selection may depend on predicted agent motion and the plan time margin statistic, the selection step must be included in the mapping that is subsequently calibrated. Let 
\begin{equation} 
\tau^\star = \sigma(I) 
\end{equation} 
denote the trajectory returned by the pipeline at plan time. 
Here, $\sigma$ represents the fixed procedure that generates candidate trajectories, evaluates them using predicted information, and selects one trajectory for execution. The procedure may depend on the candidate set, predicted agent futures, the statistic $\widehat{m}$, map and route information, and any other quantities available at planning time, but it does not use the realized future $O^\star$. 
The certification result does not require a particular form of selector or ranking objective. However, once a selection procedure is chosen, the complete procedure must be fixed before calibration and used unchanged at deployment. 
In particular, proposal generation, eligibility rules, prediction, risk evaluation, ranking objectives, and fallback behavior that affect the selected trajectory are treated as part of this fixed mapping. The relevant output at plan time is therefore 
\begin{equation} 
I \longmapsto \left( \tau^\star,\, \widehat{m}(\tau^\star,I) \right). \end{equation} 
Conformal calibration is applied to the residual associated with this selected output, rather than to candidate trajectories independently. Consequently, the resulting certificate applies directly to the trajectory produced by the fixed planning pipeline.

\subsection{Post-Selection Conformal Calibration}

Post-selection calibration accounts for the discrepancy between the plan time margin assigned to the selected trajectory and the margin realized under the observed future. For window $w$ in calibration unit $U_j$, let
\begin{equation}
    \tau_{j,w} = \sigma(I_{j,w})
\end{equation}
denote the trajectory returned by the fixed pipeline. Defining the post selection residual as
\begin{equation}
    s_{j,w}
    =
    \widehat{m}(\tau_{j,w}, I_{j,w})
    -
    h_c(\tau_{j,w}, O^\star_{j,w}),
    \label{eq:window_residual}
\end{equation}
which measures how much statistic overestimates the realized smooth geometric margin at the plan time. Larger values corresponding to greater overestimation.

Exchangeability is assumed across driving sessions rather than across planning windows. To permit arbitrary dependence among windows from the same session, the planning window residuals are aggregated into a single unit score,
\begin{equation}
    S_j
    =
    \max_{w \in W(U_j)} s_{j,w}.
    \label{eq:unit_score}
\end{equation}
Consequently, an upper bound on the unit level score implies the same bound on every evaluated window within that session.

Given $n$ calibration units and a target miscoverage level $\eta\in(0,1)$, the standard split conformal finite sample quantile \cite{vovk2005algorithmic} is defined by
\begin{equation}
    k
    =
    \left\lceil
        (n+1)(1-\eta)
    \right\rceil
\end{equation}
and let
\begin{equation}
    \widehat{q}
    =
    S_{(k)},
    \label{eq:conformal_quantile}
\end{equation}
where $S_{(k)}$ denotes the $k$th order statistic of $S_1,\ldots,S_n$. If $k>n$, $\widehat{q}$ is set to $+\infty$. Because the unit level scores are computed after applying the fixed selection pipeline, $\widehat{q}$ calibrates the post selection residuals of that pipeline.

For a selected trajectory, the resulting calibrated lower safety clearance certificate
is
\begin{equation}
    \widehat{c}(\tau,I)
    =
    \widehat{m}(\tau,I)-\widehat{q}.
    \label{eq:clearance_certificate}
\end{equation}

The following result shows that this quantity lower bounds the realized exact OBB Safety clearance with finite sample coverage.

\begin{theorem}[Exact Safety Clearance Bound]
\label{thm:exact_clearance}
Under Assumption~\ref{ass:pipeline} and~\ref{ass:exchangeability}, for a fresh exchangeable unit $U_{n+1}$,
\begin{equation}
\mathbb{P}\!\left(
    \forall w \in W(U_{n+1}):
    d^\star(\tau_w,O^\star_w)
    \ge
    \widehat{m}(\tau_w,I_w)-\widehat{q}
\right)
\ge
1-\eta .
\label{eq:exact_clearance_bound}
\end{equation}
\end{theorem}

\begin{proof}
For the fresh unit $U_{n+1}$, define its unit-level score analogously as
\begin{equation}
    S_{n+1}
    =
    \max_{w\in W(U_{n+1})}
    \left[
        \widehat{m}(\tau_w,I_w)
        -
        h_c(\tau_w,O^\star_w)
    \right].
\end{equation}

Under Assumption~\ref{ass:pipeline} and Assumption~\ref{ass:exchangeability}, the same fixed score construction is applied to each exchangeable unit. Hence, $S_1,\ldots,S_n,S_{n+1}$ are exchangeable. By the standard split-conformal rank argument \cite{vovk2005algorithmic} and the finite-sample order statistic in \eqref{eq:conformal_quantile},
\begin{equation}
    \mathbb{P}
    \left(
        S_{n+1}\le\widehat{q}
    \right)
    \ge
    1-\eta .
\end{equation}

On the event $S_{n+1}\le\widehat{q}$, every
$w\in W(U_{n+1})$ satisfies
\begin{equation}
    \widehat{m}(\tau_w,I_w)
    -
    h_c(\tau_w,O^\star_w)
    \le
    \widehat{q},
\end{equation}
or equivalently,
\begin{equation}
    h_c(\tau_w,O^\star_w)
    \ge
    \widehat{m}(\tau_w,I_w)-\widehat{q}.
\end{equation}
Lemma~1 gives the pointwise geometric relation
\begin{equation}
    d^\star(\tau_w,O^\star_w)
    \ge
    h_c(\tau_w,O^\star_w).
\end{equation}
Combining the two inequalities yields
\eqref{eq:exact_clearance_bound}. Because the score in
\eqref{eq:unit_score} is the maximum over all windows in the unit, the bound
holds simultaneously for every evaluated window in $U_{n+1}$.
\end{proof}

Theorem~\ref{thm:exact_clearance} separates predictive informativeness from statistical validity. The predictor and the choice of $\widehat{m}$ determine the tightness of the certificate, whereas the finite sample coverage guarantee is supplied by conformal calibration. In particular, $\widehat{m}$ need not itself be a valid estimate or lower bound on realized safety clearance.

For deployment, the lower safety clearance certificate can be converted into a direct acceptance rule for any requested safety clearance buffer $\rho\ge0$.

\begin{corollary}[Safety Clearance Certification Guarantee]
\label{cor:false_certification}
For a requested safety clearance buffer $\rho\ge0$, certify a selected trajectory at window $w$ whenever
\begin{equation}
    \widehat{m}(\tau_w,I_w)
    \ge
    \widehat{q}+\rho.
    \label{eq:certification_rule}
\end{equation}
Then
\begin{equation}
\mathbb{P}\!\left(
    \exists w\in W(U_{n+1}):
    \widehat{m}(\tau_w,I_w)\ge\widehat{q}+\rho,
    \;
    d^\star(\tau_w,O^\star_w)<\rho
\right)
\le
\eta .
\label{eq:false_certification}
\end{equation}
\end{corollary}

\begin{proof}
Suppose a false certification occurs for some
$w\in W(U_{n+1})$. Then
\begin{equation}
    d^\star(\tau_w,O^\star_w)
    <
    \rho
    \le
    \widehat{m}(\tau_w,I_w)-\widehat{q}.
\end{equation}
By Lemma~1,
\begin{equation}
    h_c(\tau_w,O^\star_w)
    \le
    d^\star(\tau_w,O^\star_w),
\end{equation}
and therefore
\begin{equation}
    \widehat{m}(\tau_w,I_w)
    -
    h_c(\tau_w,O^\star_w)
    >
    \widehat{q}.
\end{equation}
Hence $S_{n+1}>\widehat{q}$. The false-certification event is therefore
contained in the split-conformal failure event, whose probability is at most
$\eta$.
\end{proof}

Equation~\eqref{eq:false_certification} controls the joint event that a trajectory is certified and its realized exact safety clearance is below the requested buffer:
\begin{equation}
\mathbb{P}\left(
\text{certified}
\cap
{d^\star<\rho}
\right)
\le
\eta.
\end{equation}
Thus, the guarantee concerns the probability of a false certification event under the deployment population. It is a marginal joint event guarantee, not a conditional guarantee  the subset of certified trajectories. Likewise, guarantees conditional on city, scenario class, or other deployment subpopulations require the corresponding exchangeability assumption to hold for those populations.

Finally, because trajectory selection is included in the residual \eqref{eq:window_residual}, the certificate applies to the trajectory returned by the fixed planning pipeline. If simultaneous certification of all candidates is required, the window level score can instead be defined as the maximum residual over $\tau\in\mathcal{C}(I_w)$ before applying the unit level maximum in \eqref{eq:unit_score}.

\section{nuPlan Experiments}

\subsection{Experimental Protocol}

Experiments are conducted on the nuPlan benchmark using driving sessions from Boston, Pittsburgh, and Singapore. The frozen dataset contains 300 independent drive sessions, with 100 sessions from each city. One planning query is evaluated per session, so each session forms one exchangeability unit and the unit level score in \eqref{eq:unit_score} reduces to a single window residual.

For each query, the native PDM-Closed proposal generator produces a bank of 15 candidate ego trajectories over a 4~s horizon. Trajectory selection is completed before accessing future annotations. The prediction module uses at most 12 dynamic agents available at planning time. In contrast, realized safety clearance is evaluated retrospectively at every horizon step over all logged vehicles, pedestrians, and bicycles, including agents omitted by the predictor after the planning query. Thus, prediction omissions are reflected in the realized margin and subsequently in the conformal residual.

The 300 sessions are divided by a seeded random split into 150 calibration sessions and 150 held-out test sessions, with no session appearing in both sets. This pooled split represents an approximately balanced three-city deployment mixture. Only queries with at least one nearby dynamic agent available for plan-time prediction are retained.

\textbf{Parameters Settings:}
Two plan-time statistics are evaluated. The nominal method uses $\widehat{m}_{\mathrm{nom}}$ in \eqref{eq:nominal_margin}, whereas the risk-aware method uses the sampled lower tail CVaR statistic $\widehat{m}_{\mathrm{cvar}}$ in \eqref{eq:sampled_cvar}. For CVaR, $\beta\in\{0.30,0.10\}$ denotes the fraction of the lower tail averaged in the risk statistic, with smaller $\beta$ placing greater emphasis on adverse predicted futures. Each statistic is calibrated independently using the same frozen calibration and test split.

Conformal calibration is evaluated at target miscoverage levels
$\eta\in\{0.10,0.05\}$, corresponding to target coverage levels of $1-\eta\in\{0.90,0.95\}$. The parameter $\rho$ denotes the requested safety clearance buffer in the certification rule
$\widehat{m}(\tau,I)-\widehat{q}\ge\rho$. Unless otherwise stated, $\rho=0$, so certification requires a nonnegative calibrated lower safety clearance bound.

\textbf{Evaluation Metrics:}
The primary metrics are the conformal correction $\widehat{q}$, exact safety clearance coverage, certificate rate, and unit level false certification rate. Exact safety clearance coverage is the fraction of test units satisfying the event in Theorem~\ref{thm:exact_clearance}. The certificate rate is the fraction of evaluated test windows satisfying $\widehat{m}(\tau,I)-\widehat{q}\ge\rho$. The unit level false certification rate is the fraction of test units containing at least one certified window whose realized exact safety clearance is below $\rho$.
\vspace{-4pt}

\subsection{Conformal Certification Results}

Table~\ref{tab:conformal-results} reports the primary combined three-city evaluation, while Table~\ref{tab:city-results} examines the same certification procedure using separate calibration and test splits within each city.

\textbf{Combined City Evaluation:}
Table~\ref{tab:conformal-results} summarizes the results on the three city split. At $\eta=0.10$, the nominal statistic requires a conformal correction of $\widehat{q}=1.43$~m and achieves an exact safety clearance coverage of $93.3\%$, with $68.7\%$ of the test trajectories receiving a nonnegative safety clearance certificate. Replacing the nominal statistic with the sampled lower tail CVaR statistic substantially reduces the required correction. For both $\beta=0.30$ and $\beta=0.10$, the correction decreases to $\widehat{q}=0.03$~m, while the certificate rate increases to $87.3\%$. The corresponding exact safety clearance coverages are $96.0\%$ and $96.7\%$,
respectively.
The same trend is observed at the stricter miscoverage level $\eta=0.05$. The nominal correction increases to $3.66$~m, reducing the certificate rate to $40.7\%$, whereas the CVaR statistics require a smaller correction of $0.24$~m and retain certificate rates of $84.0$--$86.0\%$.

Across all settings, the empirical exact safety clearance coverage remains at or above
the corresponding target level. These results indicate that the choice of
plan time statistic primarily affects certificate tightness and availability,
while finite sample validity is provided by the conformal calibration.

\begin{table}[h]
\label{tab:conformal-results}
\caption{Session level conformal evaluation on the nuPlan PDM benchmark. Exact cov. denotes the event in Theorem~\ref{thm:exact_clearance}; Unit FC is the unit level false certification rate. Certification uses $\rho=0$.}
\vspace{-10pt}
\centering
\footnotesize
\begin{tabular}{lccccc}
\toprule
Statistic & \(\eta\) & \(\qhat\) & Exact cov. & Cert. rate & Unit FC \\
\midrule
Nominal & .10 & 1.43 & .933 & .687 & .007 \\
CVaR .30 & .10 & 0.03 & .960 & .873 & .020 \\
CVaR .10 & .10 & 0.03 & .967 & .873 & .020 \\
Nominal & .05 & 3.66 & .967 & .407 & .007 \\
CVaR .30 & .05 & 0.24 & .973 & .860 & .020 \\
CVaR .10 & .05 & 0.24 & .973 & .840 & .013 \\
\bottomrule
\end{tabular}
\end{table}

\begin{figure}[h]
\centering
\includegraphics[width=0.4\textwidth]{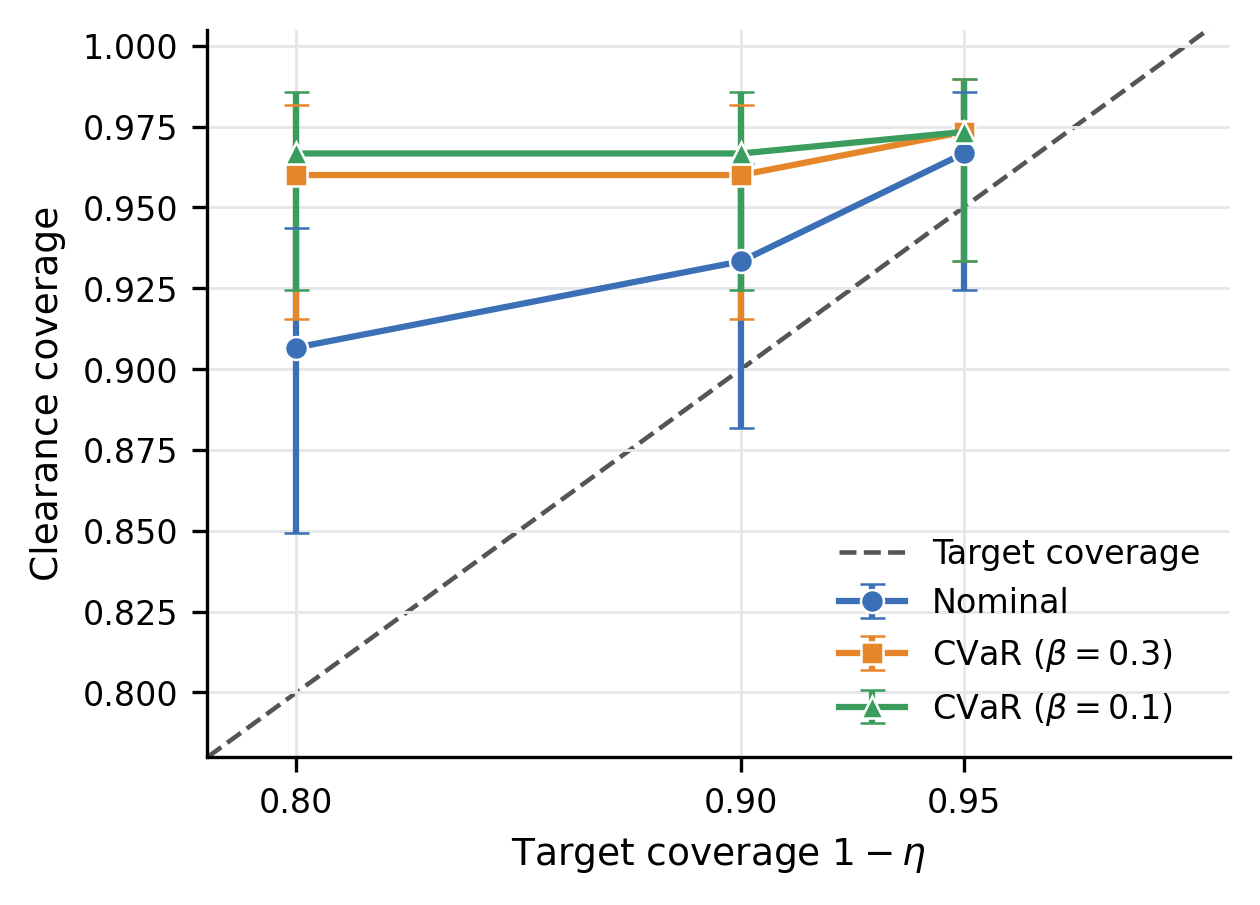}
\vspace{-12pt}
\caption{Exact safety clearance coverage versus target coverage $1-\eta$ over 150 held-out test sessions. Error bars denote 95\% Wilson intervals, and the dashed line indicates the target coverage.}
\label{fig:validity}
\end{figure}
\vspace{-8pt}

Figure~\ref{fig:validity} shows that safety clearance coverage tracks the target level across all evaluated statistics. The CVaR-based statistics remain consistently above the target and exhibit higher coverage than the nominal statistic at the lower target levels. The observed coverage is consistent with the finite sample validity guarantee. Differences in empirical coverage across the statistics reflect finite sample behavior.

\textbf{Within City Diagnostic:}
Table~\ref{tab:city-results} reports separate city calibration and results at $\eta=0.10$ and $\rho=0$, using 50 calibration and 50 test sessions per city. The nominal conformal correction varies substantially across cities, from $\widehat{q}=0.15$~m in Pittsburgh to $\widehat{q}=2.70$~m in Singapore, with the Singapore certificate rate decreasing to $60\%$. The larger Singapore correction is caused by several large nominal residuals in the upper tail of its frozen calibration split. Because this larger correction directly lowers $\widehat{m}-\widehat{q}$, the Singapore certificate rate decreases to $60\%$, compared with $88\%$ in Boston and $92\%$ in Pittsburgh.

For the CVaR statistic with $\beta=0.10$, the correction remains between $0.02$ and $0.24$~m across the three cities, with exact safety clearance coverage of $96$--$98\%$ and certificate rates of $88$--$90\%$. The smaller variation suggests greater empirical stability of the CVaR-based statistic under these frozen within city splits. The city specific results are interpreted as diagnostics of finite-sample calibration behavior rather than as a ranking of city difficulty or as additional conditional conformal guarantees.

\begin{table}[h]
\label{tab:city-results}
\caption{Within-city conformal evaluation at $\eta=0.10$ and $\rho=0$ using 50 calibration and 50 test sessions per city. Cov. denotes exact safety clearance coverage and Cert. the certificate rate.}
\vspace{-10pt}
\centering
\footnotesize
\setlength{\tabcolsep}{3pt}
\begin{tabular}{lrrrrrr}
\toprule
& \multicolumn{3}{c}{Nominal} & \multicolumn{3}{c}{CVaR .10} \\
\cmidrule(lr){2-4}\cmidrule(lr){5-7}
City & \(\qhat\) & Cov. & Cert. & \(\qhat\) & Cov. & Cert. \\
\midrule
Boston     & .66 & .88 & .88 & .24 & .98 & .90 \\
Pittsburgh & .15 & .92 & .92 & .03 & .96 & .90 \\
Singapore  & 2.70 & .94 & .60 & .02 & .98 & .88 \\
\bottomrule
\end{tabular}
\end{table}
\vspace{-10pt}

\begin{figure*}[t]
\centering
\includegraphics[width=0.85\textwidth]{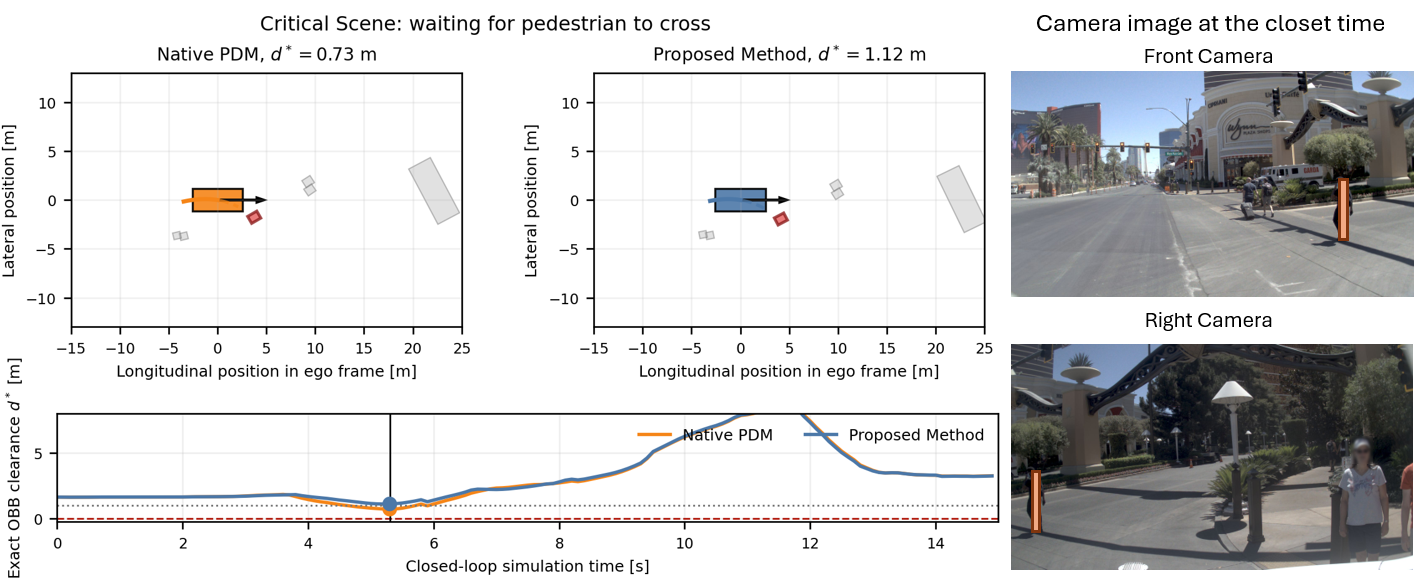}
\vspace{-11pt}
\caption{Illustrative paired CLS-NR pedestrian scene. The top panels compare native PDM and the proposed method at an intervention frame; the closest dynamic agent is shown in red, with solid and dashed paths indicating subsequent motion. At the illustrated frame, the exact OBB safety clearance is $0.73$~m for native PDM and $1.12$~m for the proposed method. The lower panel shows the safety clearance traces over the complete rollout. The vertical line marks the illustrated frame, while the horizontal dashed lines mark contact at $0$~m and the $1$~m clearance threshold.}
\label{fig:pedestrian-case}
\end{figure*}
\vspace{-4pt}
\begin{figure*}[t]
\centering
\includegraphics[width=0.7\textwidth]{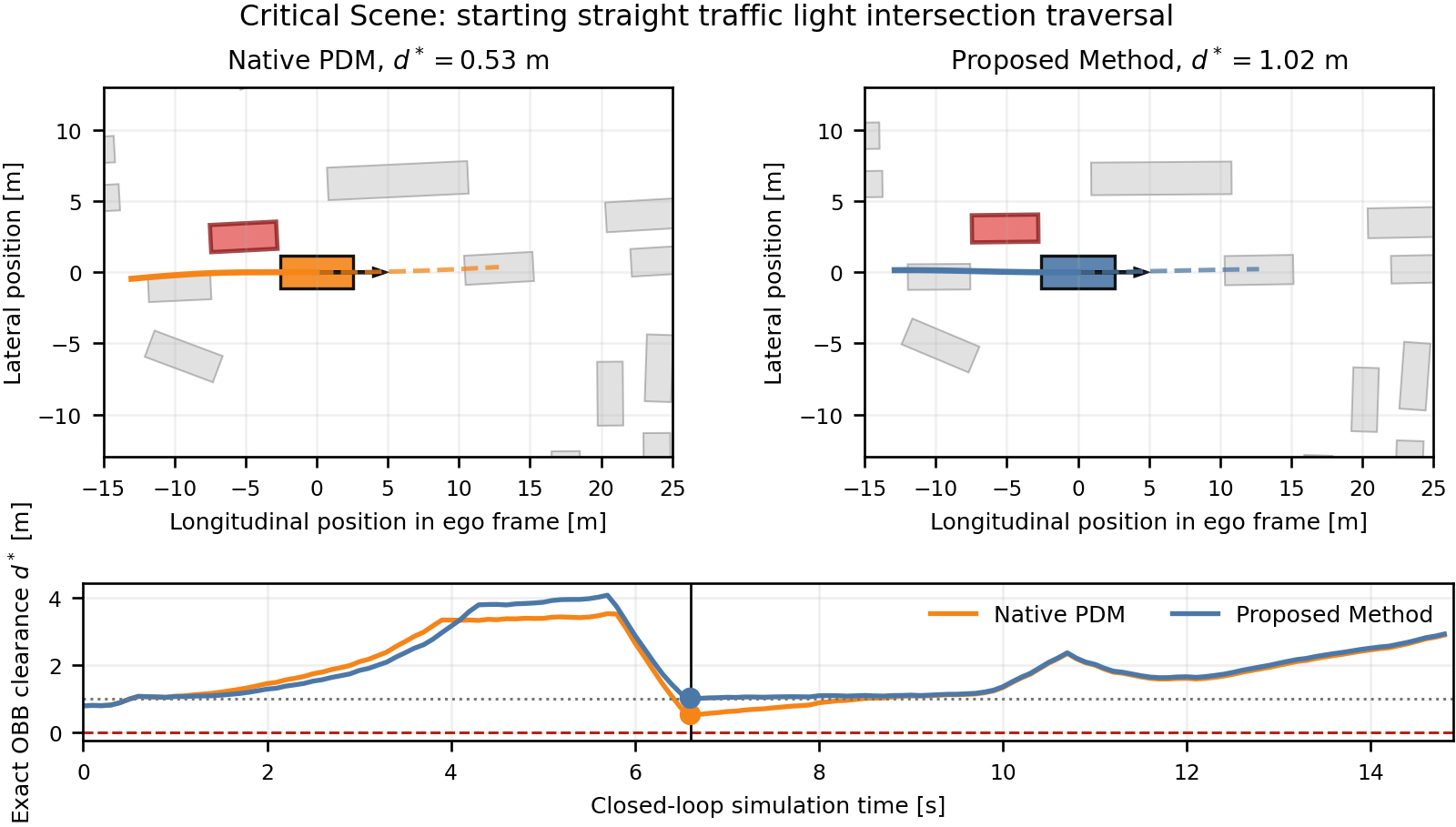}
\vspace{-11pt}
\caption{Illustrative paired CLS-NR traffic-light traversal. The top panels compare native PDM and the proposed method at the native PDM minimum safety clearance; the proposed trajectory has already diverged following an earlier intervention. At the illustrated frame, the exact OBB safety clearance is $0.53$~m for native PDM and $1.02$~m for the proposed method.}
\label{fig:traffic-light-case}
\end{figure*}
\subsection{Offline Comparison with PDM Native Selection}
This offline comparison examines how often the proposed method and PDM's native selector~\cite{dauner2023pdm} choose different trajectories from the same 15-proposal bank. This selection does not aim to reproduce PDM's native composite objective; instead, it is based on the safety clearance plan time statistic defined in Section~\ref{sec:riskaware}.

The certificate is a bound and acceptance rule, so there is no intrinsic proposal to compare with PDM. To quantify the difference between the two selection procedures, the selected trajectory index is compared with PDM's native selected index on the same 15-proposal bank. In the 300 session offline benchmark, their indices differ in 271/300 (90.3\%) nominal, 268/300 (89.3\%) CVaR .30, and 267/300 (89.0\%) CVaR .10 queries.  For nominal selection the city specific rates are 95\%, 79\%, and 97\% in Boston, Pittsburgh, and Singapore; for CVaR .10 they are 93\%, 79\%, and 95\%.  This high disagreement is expected because the proposed method uses a safety clearance gate followed by progress and effort, whereas native PDM maximizes its own composite score. It measures objective difference, not PDM error and not certificate invalidity.
\vspace{-4pt}

\subsection{Closed-Loop Integration Diagnostic}

The proposed method is further evaluated in closed-loop simulation as a lightweight intervention on PDM-Closed. At each replanning step, the native PDM proposal is retained when its predicted margin is nonnegative. Otherwise, the highest score alternative with a nonnegative margin is selected. If no such alternative exists, the native proposal is retained and the margin violation is recorded. The evaluation contains 50 paired CLS-NR scenarios from 30 drive sessions.

Across 5409 replanning windows, the native PDM proposal is retained in 5043 windows ($93.2\%$), and 27 of the 50 scenarios require no override. The mean official nuPlan score is 0.9403 for the proposed method and 0.9437 for native PDM. Aggregate collision, drivable-area, TTC, and comfort metrics remain unchanged, while the small score reduction is primarily associated with route progress. Rresults indicate that the safety clearance based intervention can be incorporated with limited disruption to the native planner behavior.

Figures~\ref{fig:pedestrian-case} and~\ref{fig:traffic-light-case} provide two closed loop examples. In the pedestrian scene, an intervention increases the illustrated exact OBB safety clearance from $0.73$ to $1.12$~m, while both rollouts retain the same official scenario score. In the traffic-light scene, an earlier intervention changes the subsequent trajectory and increases the illustrated spacing from $0.53$ to $1.02$~m. The corresponding full rollout minimum safety clearances are $0.533$ and $0.795$~m for native PDM and the proposed method, respectively. These examples illustrate how the selection rule can alter realized spacing, but they are qualitative diagnostics rather than certificate validity evidence.

Table~\ref{tab:closed_loop} reports executed OBB safety clearance against all simulated dynamic agents. Compared with native PDM, the proposed method has a mean minimum safety clearance of 1.238 versus 1.223~m and below 1~m exposure of 0.2377 versus 0.2497. Paired bootstrap intervals over the 30 drive session clusters generally include zero, so these differences are treated as descriptive integration diagnostics rather than evidence of improved closed loop safety.

The closed-loop variant uses an uncalibrated zero-margin gate and is therefore not covered by the conformal guarantee in Theorem~\ref{thm:exact_clearance}. Extending the guarantee to closed-loop operation would require calibration of the complete frozen rollout policy over disjoint drive sessions while accounting for dependent replanning windows within each session. No closed-loop conformal guarantee is claimed in this study.

\begin{table}[h]
\caption{Safety clearance diagnostic for 50 scenarios from 30 drive sessions. Exposure denotes the fraction of frames below the indicated exact OBB safety clearance.}
\label{tab:closed_loop}
\vspace{-10pt}
\centering
\footnotesize
\begin{tabular}{lrrr}
\toprule
Metric & Native PDM & Proposed & Difference \\
\midrule
Minimum clearance [m] & 1.223 & 1.238 & +0.016 \\
5th-percentile clearance [m] & 1.603 & 1.625 & +0.022 \\
Overlap exposure & .0277 & .0255 & -.0023 \\
Exposure below 0.5\,m & .0591 & .0562 & -.0029 \\
Exposure below 1.0\,m & .2497 & .2377 & -.0120 \\
\bottomrule
\end{tabular}
\end{table}
\vspace{-10pt}

\section{Conclusion}

This work presents a safety clearance certification framework for driving trajectories. A deterministic SAT margin connects the conformal certificate to exact OBB safety clearance, while session calibration handles dependence among planning windows. On the nuPlan evaluation, sampled lower tail CVaR produces substantially tighter and more available certificates than nominal prediction while maintaining the target coverage.

The additional city and closed-loop studies are diagnostic rather than separate conformal guarantees. Results show that calibration tightness can depend on the deployment population and that the safety clearance selection rule can be integrated with PDM Closed with limited intervention to its native behavior. The guarantee remains marginal under the stated exchangeability and eligibility assumptions; extending it to closed-loop deployment requires calibration of the complete frozen rollout policy.
\vspace{-8pt}

\bibliographystyle{IEEEtran}
\bibliography{references}

@inproceedings{dauner2023pdm,
  author    = {Daniel Dauner and Marcel Hallgarten and Andreas Geiger and Kashyap Chitta},
  title     = {Parting with Misconceptions about Learning-Based Vehicle Motion Planning},
  booktitle = {Proceedings of the Conference on Robot Learning},
  year      = {2023}
}

@InProceedings{fan2025riskaware,
author="Fan, Yixuan
and Li, Yali
and Wang, Shengjin",
editor="Leonardis, Ale{\v{s}}
and Ricci, Elisa
and Roth, Stefan
and Russakovsky, Olga
and Sattler, Torsten
and Varol, G{\"u}l",
title="Risk-Aware Self-consistent Imitation Learning for Trajectory Planning in Autonomous Driving",
booktitle="Computer Vision -- ECCV 2024",
year="2025",
publisher="Springer Nature Switzerland",
address="Cham",
pages="270--287",
}

@misc{zhang2025carplanner,
      title={CarPlanner: Consistent Auto-regressive Trajectory Planning for Large-scale Reinforcement Learning in Autonomous Driving}, 
      author={Dongkun Zhang and Jiaming Liang and Ke Guo and Sha Lu and Qi Wang and Rong Xiong and Zhenwei Miao and Yue Wang},
      year={2025},
      eprint={2502.19908},
      archivePrefix={arXiv},
      primaryClass={cs.RO},
      url={https://arxiv.org/abs/2502.19908}, 
}

@misc{song2025momad,
      title={Don't Shake the Wheel: Momentum-Aware Planning in End-to-End Autonomous Driving}, 
      author={Ziying Song and Caiyan Jia and Lin Liu and Hongyu Pan and Yongchang Zhang and Junming Wang and Xingyu Zhang and Shaoqing Xu and Lei Yang and Yadan Luo},
      year={2025},
      eprint={2503.03125},
      archivePrefix={arXiv},
      primaryClass={cs.RO},
      url={https://arxiv.org/abs/2503.03125}, 
}

@misc{fan2025sahdrivescenarioawarehybridplanner,
      title={SAH-Drive: A Scenario-Aware Hybrid Planner for Closed-Loop Vehicle Trajectory Generation}, 
      author={Yuqi Fan and Zhiyong Cui and Zhenning Li and Yilong Ren and Haiyang Yu},
      year={2025},
      eprint={2505.24390},
      archivePrefix={arXiv},
      primaryClass={cs.RO},
      url={https://arxiv.org/abs/2505.24390}, 
}

@article{ALLAMAA2024392,
title = {Real-time MPC with Control Barrier Functions for Autonomous Driving using Safety Enhanced Collocation},
journal = {IFAC-PapersOnLine},
volume = {58},
number = {18},
pages = {392-399},
year = {2024},
note = {8th IFAC Conference on Nonlinear Model Predictive Control NMPC 2024},
issn = {2405-8963},
doi = {https://doi.org/10.1016/j.ifacol.2024.09.058},
url = {https://www.sciencedirect.com/science/article/pii/S240589632401437X},
author = {Jean Pierre Allamaa and Panagiotis Patrinos and Toshiyuki Ohtsuka and Tong Duy Son},
}

@misc{lin2025modelbasedpolicyadaptationclosedloop,
      title={Model-Based Policy Adaptation for Closed-Loop End-to-End Autonomous Driving}, 
      author={Haohong Lin and Yunzhi Zhang and Wenhao Ding and Jiajun Wu and Ding Zhao},
      year={2025},
      eprint={2511.21584},
      archivePrefix={arXiv},
      primaryClass={cs.RO},
      url={https://arxiv.org/abs/2511.21584}, 
}

@article{RENGANATHAN2025472,
title = {Experimental Evaluation of Model Predictive Control and Control Barrier Functions for Autonomous Driving Obstacle Avoidance},
journal = {IFAC-PapersOnLine},
volume = {59},
number = {19},
pages = {472-477},
year = {2025},
note = {13th IFAC Symposium on Nonlinear Control Systems NOLCOS 2025},
issn = {2405-8963},
doi = {https://doi.org/10.1016/j.ifacol.2025.11.079},
url = {https://www.sciencedirect.com/science/article/pii/S2405896325017422},
author = {Vishnu Renganathan and Pei Yu Chang and Qadeer Ahmed},
}

@misc{chen2025dynamichighordercontrolbarrier,
      title={Dynamic High-Order Control Barrier Functions with Diffuser for Safety-Critical Trajectory Planning at Signal-Free Intersections}, 
      author={Di Chen and Ruiguo Zhong and Kehua Chen and Zhiwei Shang and Meixin Zhu and Edward Chung},
      year={2025},
      eprint={2412.00162},
      archivePrefix={arXiv},
      primaryClass={cs.RO},
      url={https://arxiv.org/abs/2412.00162}, 
}

@misc{chang2026riskbudgeted,
      title={Risk-Budgeted Control Framework for Balanced Performance and Safety in Autonomous Vehicles}, 
      author={Pei Yu Chang and Vishnu Renganathan and Qadeer Ahmed},
      year={2026},
      eprint={2510.10442},
      archivePrefix={arXiv},
      primaryClass={eess.SY},
      url={https://arxiv.org/abs/2510.10442}, 
}

@article{CHANG2026107114,
title = {Risk aware safe control with multi-modal sensing for dynamic obstacle avoidance},
journal = {Control Engineering Practice},
volume = {175},
pages = {107114},
year = {2026},
issn = {0967-0661},
doi = {https://doi.org/10.1016/j.conengprac.2026.107114},
url = {https://www.sciencedirect.com/science/article/pii/S0967066126003576},
author = {Pei Yu Chang and Qizhe Xu and Vishnu Renganathan and Qadeer Ahmed},
}

@article{Mustafa2025riskaware,
   title={RACP: Risk-Aware Contingency Planning With Multi-Modal Predictions},
   volume={10},
   ISSN={2379-8858},
   url={http://dx.doi.org/10.1109/TIV.2024.3411530},
   DOI={10.1109/tiv.2024.3411530},
   number={1},
   journal={IEEE Transactions on Intelligent Vehicles},
   publisher={Institute of Electrical and Electronics Engineers (IEEE)},
   author={Mustafa, Khaled A. and Ornia, Daniel Jarne and Kober, Jens and Alonso-Mora, Javier},
   year={2025},
   month=Jan, pages={228–243} }

@article{raeesi2025collision,
  author  = {Raeesi, H. and Khosravi, A. and Sarhadi, P.},
  title   = {Collision Avoidance for Autonomous Vehicles Using Reachability-Based Trajectory Planning in Highway Driving},
  journal = {Proceedings of the Institution of Mechanical Engineers, Part D: Journal of Automobile Engineering},
  volume  = {239},
  number  = {4},
  pages   = {1003--1020},
  year    = {2025},
  doi     = {10.1177/09544070231222053}
}

@misc{chakraborty2025safety,
      title={Safety Evaluation of Motion Plans Using Trajectory Predictors as Forward Reachable Set Estimators}, 
      author={Kaustav Chakraborty and Zeyuan Feng and Sushant Veer and Apoorva Sharma and Wenhao Ding and Sever Topan and Boris Ivanovic and Marco Pavone and Somil Bansal},
      year={2025},
      eprint={2507.22389},
      archivePrefix={arXiv},
      primaryClass={cs.RO},
      url={https://arxiv.org/abs/2507.22389}, 
}

@InProceedings{pmlrv283contreras25a,
  title = 	 {Safe, Out-of-Distribution-Adaptive MPC with Conformalized Neural Network Ensembles},
  author =       {Contreras, Jose Leopoldo and Shorinwa, Ola and Schwager, Mac},
  booktitle = 	 {Proceedings of the 7th Annual Learning for Dynamics \&amp; Control Conference},
  pages = 	 {194--207},
  year = 	 {2025},
  editor = 	 {Ozay, Necmiye and Balzano, Laura and Panagou, Dimitra and Abate, Alessandro},
  volume = 	 {283},
  series = 	 {Proceedings of Machine Learning Research},
  month = 	 {04--06 Jun},
  publisher =    {PMLR},
  url = 	 {https://proceedings.mlr.press/v283/contreras25a.html},
}

@misc{doula2025safepathconformalpredictionsafe,
      title={SafePath: Conformal Prediction for Safe LLM-Based Autonomous Navigation}, 
      author={Achref Doula and Max Mühlhäuser and Alejandro Sanchez Guinea},
      year={2025},
      eprint={2505.09427},
      archivePrefix={arXiv},
      primaryClass={cs.LG},
      url={https://arxiv.org/abs/2505.09427}, 
}

@article{karnchanachari2024nuplan,
  author  = {Napat Karnchanachari and Dimitris Geromichalos and Kok Seang Tan and Nanxiang Li and Christopher Eriksen and Shakiba Yaghoubi and Noushin Mehdipour and Gianmarco Bernasconi and Whye Kit Fong and Yiluan Guo and Holger Caesar},
  title   = {Towards Learning-Based Planning: The nuPlan Benchmark for Real-World Autonomous Driving},
  journal = {arXiv preprint arXiv:2403.04133},
  year    = {2024}
}

@book{vovk2005algorithmic,
  author    = {Vladimir Vovk and Alexander Gammerman and Glenn Shafer},
  title     = {Algorithmic Learning in a Random World},
  publisher = {Springer},
  year      = {2005}
}

@article{lindemann2023safe,
  author  = {Lars Lindemann and Matthew Cleaveland and Gihyun Shim and George J. Pappas},
  title   = {Safe Planning in Dynamic Environments Using Conformal Prediction},
  journal = {IEEE Robotics and Automation Letters},
  volume  = {8},
  number  = {8},
  pages   = {5116--5123},
  year    = {2023},
  doi     = {10.1109/LRA.2023.3292071}
}

@inproceedings{dixit2023adaptive,
  author    = {Anushri Dixit and Lars Lindemann and Skylar X. Wei and Matthew Cleaveland and George J. Pappas and Joel W. Burdick},
  title     = {Adaptive Conformal Prediction for Motion Planning among Dynamic Agents},
  booktitle = {Proceedings of the 5th Annual Learning for Dynamics and Control Conference},
  series    = {Proceedings of Machine Learning Research},
  volume    = {211},
  pages     = {300--314},
  year      = {2023}
}

@INPROCEEDINGS{takasugi2024SSAT,
  author={Takasugi, Noriaki and Kinoshita, Masaya and Kamikawa, Yasuhisa and Tsuzaki, Ryoichi and Sakamoto, Atsushi and Kai, Toshimitsu and Kawanami, Yasunori},
  booktitle={2024 IEEE/RSJ International Conference on Intelligent Robots and Systems (IROS)}, 
  title={Real-time Perceptive Motion Control using Control Barrier Functions with Analytical Smoothing for Six-Wheeled-Telescopic-Legged Robot Tachyon 3}, 
  year={2024},
  volume={},
  number={},
  pages={6802-6809},
  doi={10.1109/IROS58592.2024.10802375}}

\end{document}